\documentclass[journal]{IEEEtran}
\ifCLASSINFOpdf
\else
\fi
\usepackage{graphicx}
\usepackage{amsmath}
\usepackage{amssymb}
\usepackage{amsthm}
\usepackage{hyperref}
\usepackage{algorithm}
\usepackage{algorithmicx}
\usepackage{algpseudocode}
\usepackage{subfigure}
\usepackage{multirow}
\usepackage[maxnames=4,minnames=3,firstinits=true]{biblatex}
\usepackage{caption}
\usepackage{lipsum}
\usepackage[misc]{ifsym}
\bstctlcite{IEEEexample:BSTcontrol}

\newcommand \equref[1]{Eq. (\ref{#1})}
\newcommand \figref[1]{Fig. \ref{#1}}
\newcommand \algoref[1]{Algotithm \ref{#1}}
\newcommand \chapref[1]{Section \ref{#1}}

\newtheorem{remark}{Remark}
\newtheorem{lemma}{Lemma}

\begin{document}

\title{Multi-agent Reinforcement Traffic Signal Control based on Interpretable Influence Mechanism \\ and Biased ReLU Approximation 
\thanks{This work was supported in part by the National Natural Science Foundation of China under Grant 62173113, and in part by the Science and Technology Innovation Committee of Shenzhen Municipality under Grant GXWD20231129101652001, and in part by Natural Science Foundation of Guangdong Province of China under Grant 2022A1515011584.(\textit{Corresponding author: Jun Xu.})}}

\author{Zhiyue Luo \thanks{Zhiyue Luo, Jun Xu, and Fanglin Chen are with the School of Mechanical Engineering and Automation, Harbin Institute of Technology (Shenzhen), Shenzhen 518055, China (email: xujunqgy@hit.edu.cn).},
        Jun Xu,
        Fanglin Chen
}
\maketitle



\begin{abstract}
Traffic signal control is important in intelligent transportation system, of which cooperative control is difficult to realize but yet vital.
Many methods model multi-intersection traffic networks as grids and address the problem using multi-agent reinforcement learning (RL). Despite these existing studies, there is an opportunity to further enhance our understanding of the connectivity and globality of the traffic networks by capturing the spatiotemporal traffic information with efficient neural networks in deep RL.

In this paper, we propose a novel multi-agent actor-critic framework based on an interpretable influence mechanism with a centralized learning and decentralized execution method.
Specifically, we first construct an actor-critic framework, for which the piecewise linear neural network (PWLNN), named biased ReLU (BReLU), is used as the function approximator
to obtain a more accurate and theoretically grounded approximation.
Then, to model the relationships among agents in multi-intersection scenarios, we introduce an interpretable influence mechanism based on efficient hinging hyperplanes neural network (EHHNN), which derives weights by ANOVA decomposition among agents and extracts spatiotemporal dependencies of the traffic features.
Finally, our proposed framework is validated on two synthetic traffic networks to coordinate signal control between intersections, achieving lower traffic delays across the entire traffic network compared to state-of-the-art (SOTA) performance.

\end{abstract}


\IEEEpeerreviewmaketitle





\section{Introduction}
\IEEEPARstart{T}{ransportation} is the key driving force for economic and social growth and is one of the manifestations of urban competitiveness.
With the rapid development of urbanization, traffic congestion has become a major challenge for cities around the world\cite{zhao2011computational}.
The increasing number of vehicles on the roads has led to longer travel times, increased fuel consumption, and higher levels of air pollution \cite{gupta2018low}.

Traffic signal control is usually regarded as the most significant and effective method for quick and safe transportation, which reduces traffic congestion in the urban network by adjusting the signal phase at intersections\cite{zhao2011computational}.
Generally, traffic signal control methods can be mainly divided into three main types, including fixed time control\cite{miller1963settings}, actuated control\cite{cools2013self} and adaptive control\cite{mannion2016experimental},\cite{haydari2020deep},\cite{abdulhai2003reinforcement}.
However, neither fixed time control nor actuated control methods consider long-term traffic conditions, thus cannot optimize the traffic signal phases adaptively based on real-time traffic flow. 
In contrast, adaptive control can effectively mitigate traffic congestion and enhance transportation efficiency, which is currently a research hotspot.
With the development of artificial intelligence, data-driven control methods play an increasingly important role in intelligent transportation systems.
Traffic signal control is a sequential decision problem, which can be modeled as Markov Decision Process (MDP) and solved by reinforcement learning (RL).
RL is a powerful dynamic control paradigm, making no additional assumptions on the underlying traffic transition model.
Rather than computing explicit traffic flow dynamic equations, RL learns the optimal strategy based on its experience interacting with the traffic environment.
The objective of traffic signal control is to minimize the total waiting time within the traffic network by controlling the phase of traffic signals at intersections.

Several RL methods have been applied to isolated traffic signal control \cite{genders2016using}, \cite{van2016deep}, \cite{behrendt2017deep}, significantly impacting the field of traffic signal control problem.
However, in real world, traffic networks are interrelated, and controlling a specific intersection signal will inevitably affect the traffic condition of the upstream and downstream intersections, leading to the chain reaction of the surrounding intersections.
Multi-intersection traffic network is a complex and nonlinear system that can be quite challenging to model.
The complexity arises from the need to process intricate spatiotemporal traffic flow data and address cooperative problems among agents, which are difficult to solve using a centralized approach.
Consequently, many studies have explored the application of multi-agent RL (MARL) in cooperative traffic signal control problems.
These studies aim to find a balance between centralized and decentralized training models, thus reaching an optimum strategy for all agents and minimizing the waiting time of vehicles within the traffic network.
To explore further the spatial structural dependencies among different intersections, many researchers applied graph RL in multi-intersection traffic signal control, which uses graph neural networks to learn and exploits representations of each agent and its neighborhood in the form of node embedding.

Although the works mentioned above deal with the multi-intersection traffic signal control problem through cooperative RL, the effect of coupled agents on the global performance of the traffic signal has yet to be considered explicitly.
Hence, in this paper, we introduce a novel interpretable influence mechanism using efficient hinging hyperplanes neural networks (EHHNN) \cite{xu2020efficient}, which aims to capture the spatiotemporal dependencies of traffic information to better build the relationships among neighboring intersections.
Then, we propose a multi-agent actor-critic framework with a centralized critic and decentralized actors.
The main idea is to learn the critic using our interpretable influence mechanism to coordinate the interaction between different agents.
Besides, we improve the function approximator used in the deep RL, which is a piecewise linear neural network (PWLNN), named biased ReLU (BReLU) neural network. This BReLU neural network can obtain superior performance than rectified linear units (ReLU) neural network in function approximation when reasonably dividing the piecewise linear (PWL) region \cite{liang2021biased}.
We approximate both the value function and policy function with BReLU neural network, and thus construct the PWL-actor-critic framework.
This also coincides with the conclusion that minimizing PWL functions over polyhedrons yields PWL solutions.

The technical contribution of this paper can be summarized as follows:
\begin{itemize}
    \item A novel MARL framework is proposed, in which the BReLU neural network is used as the approximator for both the value function and policy function to construct the PWL-actor-critic framework. 
    \item We propose a novel influence mechanism using the EHHNN and combine it with MARL. Compared to the graph-based approaches, our proposed mechanism does not require the pre-defined adjacency matrix of the traffic network and exhibits excellent capability in capturing the spatiotemporal dependencies of traffic flow data. Instead, it models the relationships by analyzing the impact of input variables on the output variables.
    Compared to the attention mechanism, the EHHNN-based influence mechanism has fewer parameters due to the sparsity of the EHHHNN and does not include a nonlinear activation function, which better explains the contribution of input features to a specific output variable.
    To the best of our knowledge, we are the first to use EHHNN to model multi-agent relationships and introduce a novel multi-agent framework using EHHNN-based influence mechanism.
    \item Experiments are conducted on both the traffic grid and a non-Euclidean traffic network. We compare the effectiveness of our solution with several state-of-the-art (SOTA) methods and further analyze the relation reasoning given by our influence mechanism.
\end{itemize}

The rest of this paper is organized as follows: 
\chapref{section2} briefly introduces related work.
\chapref{section3} gives problem formulation of multi-intersection traffic signal control problem and models it as Partially-Observable MDP (POMDP).
\chapref{section4} outlines the detailed implementations of our proposed multi-agent actor-critic framework, which is based on BReLU neural network approximation and employs a novel interpretable influence mechanism to learn the spatiotemporal dependency among different agents.
\chapref{chapter5} demonstrate the effectiveness of the EHHNN through traffic forecasting experiments, while also comparing our proposed multi-agent actor-critic framework with traditional fixed time control and SOTA MARL algorithm in two simulated environments of multi-intersection traffic signals control task.
Finally, \chapref{chapter6} concludes the paper and outlines the main results.

\section{Related Work} \label{section2}
In this section, we briefly introduce the related work of MARL and graph neural networks in traffic signal control.
Additionally, we introduce the development of the EHHNN and its application.
\subsection{Deep Graph Convolutional RL}
In order to tackle the challenges of large-scale traffic signal control and address the issue of the curse of dimensionality in MARL, many studies have explored the application of MARL methods to traffic signal control problems, which includes approaches employing independent deep Q-network\cite{van2016coordinated}, \cite{calvo2018heterogeneous}, multi-agent deep deterministic policy gradient \cite{casas2017deep}, multi-agent advantage actor-critic \cite{chu2019multi}, and large-scale decomposition method \cite{tan2019cooperative}.
However, not all data in the real world can be represented as a sequence or a grid.
To explore further the spatial structural dependencies among different intersections, many researchers applied graph RL in multi-intersection traffic signal control, which uses graph neural networks \cite{kipf2016semi}, \cite{hamilton2017inductive}, \cite{velivckovic2017graph} to learn and exploits representations of each agent and its neighborhood in the form of node embedding.
Then, researchers introduced structured perception and relational reasoning based on graph networks in MARL \cite{zambaldi2018relational}, \cite{battaglia2018relational}, which aims to model the relationship in multi-agent scenarios.
Following this line, graph-based RL methods have been used for traffic signal control to attain a deeper understanding of interactions among different intersections, e.g., graph convolutional network and graph attention network were applied in RL to process graph-structured traffic flow data \cite{jiang2018graph}, \cite{chen2020gama}.
Several studies combined deep Q network method with graph neural networks \cite{devailly2021ig}, \cite{yan2023graph}.
In another study, researchers introduced multi-agent advantage actor-critic algorithm to graph-based RL and proposed a decentralized graph-based method \cite{zeng2021graphlight}.
Furthermore, a spatiotemporal MARL framework was proposed, which considers the spatial dependency and temporal dependency of the traffic information \cite{wang2020stmarl}.
\subsection{The EHHNN}
In deep RL, dealing with high-dimensional and intricate state spaces is a major challenge, so many studies use an effective function approximator to approximate the value function and policy function.
Neural networks enable the integration of perception, decision-making, and execution phases, facilitating end-to-end learning.
PWLNNs are now a successful mainstream method in deep learning, and ReLU is a commonly used activation function in PWLNNs.
In typical neural networks, it is impossible to determine the contribution of different input variables to the output through neuron connections. Therefore, in the multi-agent cooperative control tasks, researchers have introduced methods incorporating attention mechanism and relational inductive biases to determine the weights of various input features through similarity measurement.
However, we aspire to capture the influence of the inputs on outputs directly through an interpretable neural network.

In 2020, a novel neural network called EHHNN was proposed \cite{xu2020efficient}, which strikes a good balance between model flexibility and interpretability.
The EHHNN is a kind of PWLNN derived from the hinging hyperplanes (HH) model \cite{breiman1993hinging}, in which the ReLU neural network is a kind of the HH model \cite{liang2021biased}. When the HH model involves only one linear hyperplane passing through the origin, it becomes a ReLU neural network.
The EHHNN exhibits a high flexibility in dynamic system identification due to its PWL characteristics. Furthermore, the network is interpretable, allowing for determining the impact of the input layer and hidden layer neurons on the output through the analysis of variance (ANOVA) \cite{friedman1991multivariate}, thereby facilitating the analysis of feature variables.
In 2021, a new activation function called BReLU is proposed \cite{liang2021biased}, which is similar to the ReLU. 
The BReLU neural network is also an HH model, i.e. a PWL function. Compared with ReLU, it employs multiple bias parameters, which can partition the input spaces into numbers of linear regions, thus resulting in high flexibility and excellent approximation capabilities, especially in regression problems.

\section{MARL for Traffic Control Problem} \label{section3}
In multi-intersection traffic signal control problems, each signal controller reduces the traffic flow at the intersection by adjusting the phase to minimize the waiting time of the traffic flow in the whole traffic network.
In this section, we first model an unstructured multi-intersection traffic network as a directed graph and give the formulation and assumptions used in the optimization problem. Then, we model the multi-intersection traffic signal control as a MDP and give a detailed formulation of the state, action, and reward function.

\subsection{Traffic Signal Control Problem Formulation} \label{modeling}
We consider a more general traffic network, as shown in \figref{TFN}.
For the urban traffic network may not be necessarily connected due to the limitation of the area for urban development \cite{jiang2004structural}.
Similar to related work \cite{van2016coordinated},\cite{devailly2021ig},\cite{mousavi2017traffic}, the typical intersection shapes are included in the traffic network, as shown in \figref{TFN}, $v_3$ and $v_5$ are three-way intersections, while others are crossroads.
Besides, the length of lanes in the traffic network varies.
We define intersections in the traffic network as nodes and the road between every two intersections as edges,
then the multi-intersection traffic network can be modeled as a directed graph $G(V,E,\Psi)$, where $V=\left \{ v_i \right \}_{i=1}^{\left | V \right | }$ is the nodes set, $\left | V \right | = N$, refers $N$ nodes (intersections) in the graph.
$E=\left \{e_j \right \} _{j=1}^{\left | E \right | }$ is the edges set, $\left | E \right |$ is the number of edges and there are $l_j$ lanes on the $j$-th edge.
For each edge $e_j$, there is a corresponding upstream node $u_j \in V$ and downstream node $d_j \in V$.
$\Psi$ represents the global attribute of the traffic graph.

\begin{figure}[!htp]
    \centering
    \includegraphics[width=8cm]{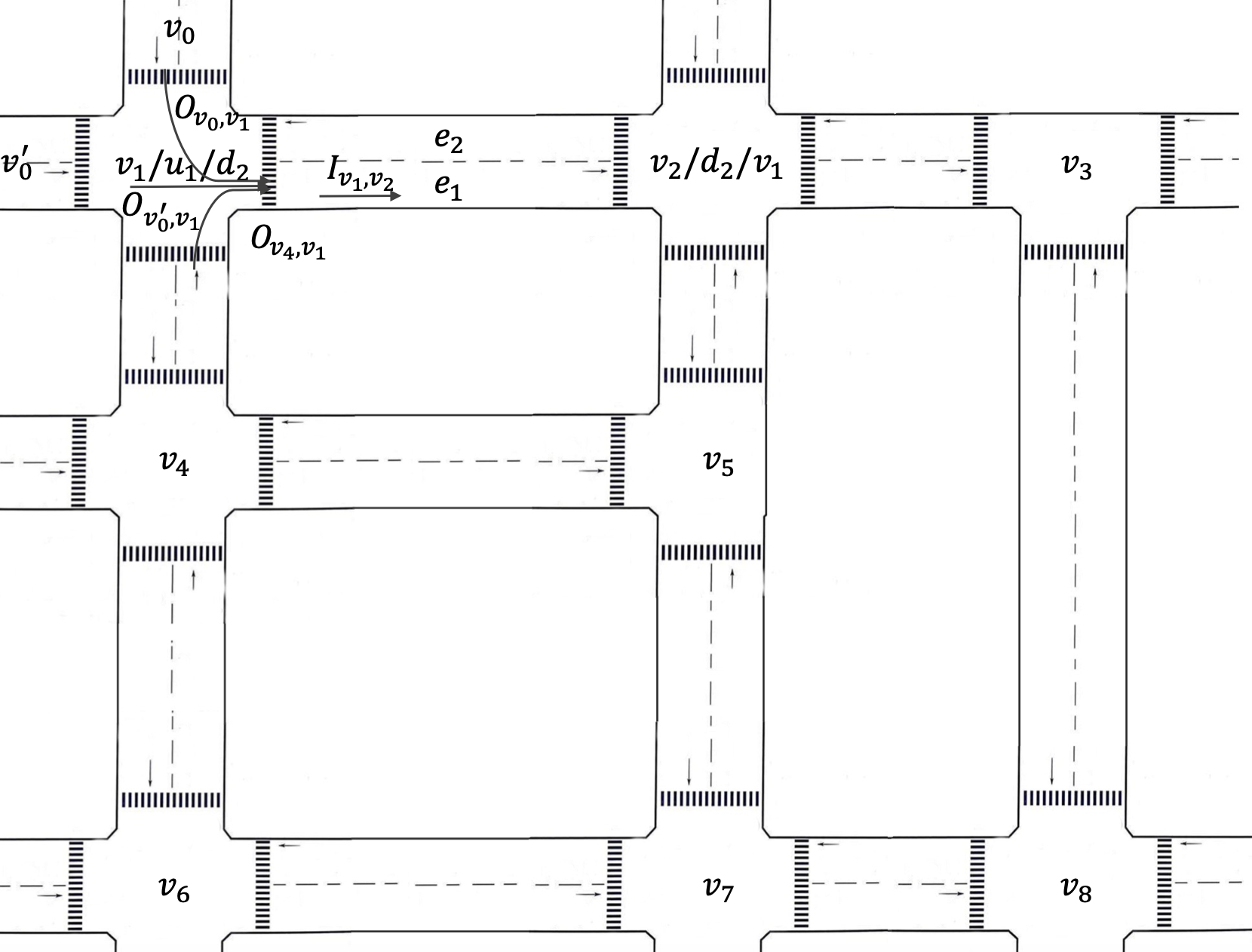}
    \caption{A non-Euclidean Traffic Network}
    \label{TFN}
\end{figure}

Before describing the traffic dynamic system, we make three reasonable assumptions.

\noindent \textbf{Assumption 1:} 
To simplify the large-scale traffic problem, we assume that the sampling time intervals $\Delta t$ of all intersections are the same, then the cycle time can be expressed as 
\begin{equation}
    T=k\cdot \Delta t
\end{equation}

\noindent \textbf{Assumption 2:} 
We assume that the vehicles from upstream node $u_j$ entering edge $e_j$ will travel freely until they reach the tail of the waiting vehicle queues. And they will be divided to join the separated queues of the downstream node $d_j$ that they intend to go to.

\noindent \textbf{Assumption 3:}
We assume that each intersection block is equipped with an individual roadside unit, which can observe and collect the traffic information and transmit it to the central controller.

Now, the dynamic traffic model can be derived as follows.
According to the vehicle conservation theorem, the number of vehicles on edge $e_j$ at step $k$ is updated by
\begin{equation}
    n_j(k+1) = n_j(k)+\left ( I_{u_j, d_j} -O_{d_j,v} \right ) ,v\in \mathcal{N}_{d_j}-\left \{ u_j \right \}
\end{equation}
where $I_{u_j,d_j}$ denotes the traffic flow entering edge $e_j$, while $O_{v,u_j}$ denotes the traffic flow leaving edge $e_j$ reaching its next adjacent target node $v$, $\mathcal{N}_{d_j}$ represents the adjacent nodes of the downstream node $d_j$.

And consequently, we can update the density of the $j$-th edge from node $u_j$ to node $d_j$ at time $k$ as
\begin{equation}
    \Phi _{u_j, d_j}=\min\left ( 1, \frac{n_j(k)\cdot \tau}{L(e_j)} \right ) 
\end{equation}
where $\tau$ is a constant representing the vehicular gap and $L(e_j)$ denotes the length of edge $j$.

The queue length $q_{u_j, d_j}(k)$ is the number of vehicles waiting on edge $e_{j}$ at step $k$, which is the number of vehicles on edge $e_{j}$ with a speed of 0. And the total queue length at intersection $i$ can be expressed as
\begin{equation}
    Q_i(k)=\sum q_{u_j,v_i}(k), u_j \in \mathcal{N}_{v_i}
\end{equation}
The total vehicle waiting time at the intersection $i$ at step $k$ is denoted as $W_i(k)$.

Our objective is to minimize the total waiting queue of the global traffic network by controlling the phase of the intersection signals
\begin{equation}
    Z^*=\min \sum_{k=0}^{T/\Delta t} \sum_{i=1}^N Q_i(k)
\end{equation}

\subsection{Multi-intersection Traffic Signal Control as MDP}\label{MDP}
Due to the uncertainty and dynamics of the traffic system, the multi-intersection traffic signal control problem can be abstracted as a discrete stochastic control problem, and be modeled as POMDP, defined as a tuple $\left \langle \mathcal{O, A, R}, N, \gamma \right \rangle$, where $\mathcal{O}=\left \{ o_1, o_2,\dots, o_N \right \} $ is the set of observations, $\mathcal{A}=\left \{ a_1, a_2, \dots , a_N \right \} $ is the set of actions, $\mathcal{R}=\left \{ r_1, r_2, \dots, r_N \right \} $ is the set of reward, and $N$ is the number of agents, also the number of intersections and the number of nodes in the graph $G$.

The observation of each agent $i$ is defined as the vector of neighborhood queue length $q_{u_j,v_i}$ of intersection $v_i$, phase of current intersection signal $\rho$, and road density $\Phi _{u_j,v_i}$, which can well represent the incoming traffic flow on the road and the queue numbers at the intersection. For agent $i$ with $u_j \in \mathcal{N}_{v_j}$ neighboring intersection, the local observation at step $k$ is
\begin{equation}
    o_i(k)=[\rho _i(k),q_{u_j, v_i}(k),\Phi _{u_j, v_i}(k)], u_j \in \mathcal{N}_{v_i}
\end{equation}
The joint state over the traffic network is expressed as $S=o_1 \times o_2 \times \dots o_N$.

Each intersection has a separate controller giving its signal (or action). At each time step $k$, the controller of intersection $i$ selects a discrete action, denoted as $a_i(k) \in \mathcal{A}_i$.
The joint action grows exponentially with the number of agents.
We consider only feasible sign configurations in the action set and use a four-stage green phase control.

In traffic signal control, researchers often use characteristic variables such as the total waiting time and the queue length of vehicles to define the reward function.
Many studies use expressions such as \equref{reward1} to define reward function, which uses changes in traffic characteristic variables $X_{tcv}$ between adjacent time step.
\begin{equation}
    r_{tcv}(k) = X_{tcv}(k-1) - X_{tcv}(k)
    \label{reward1}
\end{equation}
Using the expression above, agents tend to accumulate a certain amount of vehicles at the intersection and then release them to obtain larger rewards. This control strategy obviously does not conform to the real-world traffic application scenario.
Therefore, in this paper, we propose an improved reward function based on \equref{reward1}. Considering the characteristics of large-scale traffic network and the optimization objective of traffic control, the new reward function is shown in \equref{reward2}, where $\kappa _1 $, $\kappa_2$ and $\kappa _3$ are hyperparameters of the reward function under different traffic conditions, $\Delta Q_i(k) = Q_i(k)-Q_i(k-1)$ is the changes in queue number between adjacent time step.
\begin{equation}
    r_i(k)=\left\{
      \begin{array}{llr}
        \kappa _1, & Q_i(k)=0 \\
        -W_i(k)/\kappa _2 , & \Delta Q_i(k)>0 \\
        -\kappa _3 \cdot \Delta Q_i(k), & \Delta Q_i(k) \le 0
    \end{array}
    \right.
    \label{reward2}
\end{equation}
The global reward on the whole traffic network is a linear weighted sum of reward $r_i$ for each agent
\begin{equation}
    r(k) = \sum_{i=1}^N r_i(k)
\end{equation}

\section{Multi-agent BReLU actor-critic with interpretable influence mechanism} \label{section4}
In this section, we propose a multi-agent BReLU actor-critic framework with an interpretable influence mechanism.
We introduce a novel neural network called BReLU neural network, which offers improved function approximation for RL and constructs a PWL-actor-critic framework.
Then, we extend the PWL-actor-critic framework to the MARL algorithm and employ the interpretable influence mechanism based on EHHNN to capture the spatiotemporal dependencies among different agents.
We use a centralized training and decentralized execution method, where a joint value function is learned from the aggregated information and each actor learns its policy function based on the local observations.

\subsection{Overview}
The overview framework of our proposed method is shown in \figref{overview}.
We first build a graph that comprises traffic signal agents and subsequently propose a novel influence mechanism to extract the spatial dependencies from the input graph.
In detail, the observations of each agent $\left\{o_k^i\right\}_{i=1}^N$ go through a node embedding layer and the EHH-based mechanism to obtain the node embedding $V_{k}^{in}$ for the actor layer and the aggregation embedding for the critic layer, respectively.
The module inside the influence mechanism is shown on the right-hand side of \figref{overview}, the observations $\left\{o_k^i\right\}_{i=1}^N$ is fed into a linear transformation layer followed by an EHHNN to obtain the hidden variable $H_k$, then an ANOVA decomposition layer is designed for feature extraction to obtain the important coefficient $\sigma _m$. Finally, the aggregation embedding $V_{k}^{out}$ is obtained through weighted aggregation.
We approximate the policy function and value function in the actor-critic layers with a novel neural network named BReLU and thus construct a PWL-actor-critic framework.
Next, we will provide a detailed description of each module within the proposed framework.

\begin{figure*}[!htp]
    \centering
    \includegraphics[width=15cm]{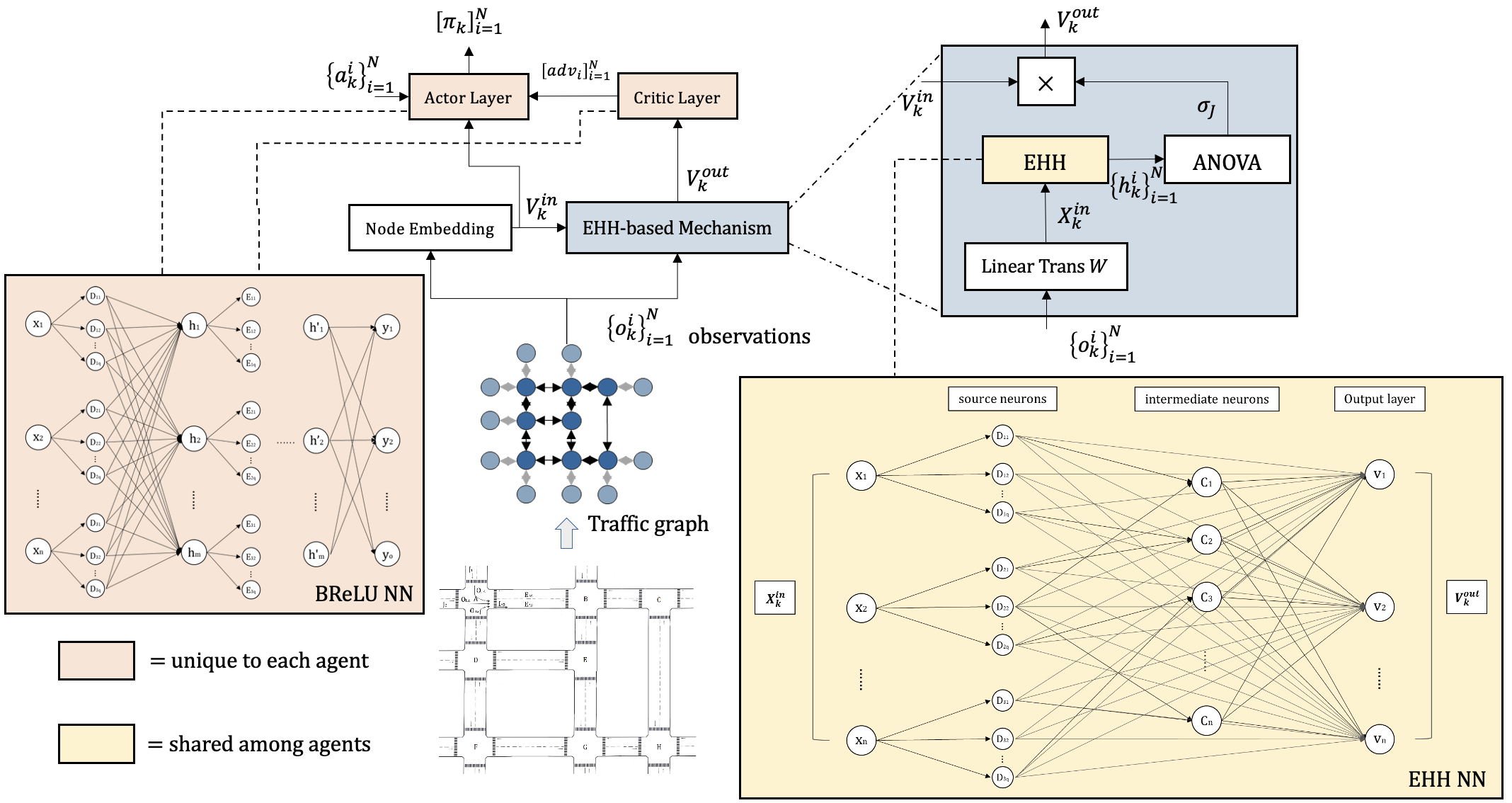}
    \caption{Overview of multi-agent BReLU actor-critic framework.}
    \label{overview}
\end{figure*}

\subsection{Node Embedding}

Firstly, we obtain the node embedding $V_k^{in}=\left \{ v_{i,k}^{in} \right \}_{i=1}^N$ at the current time step $k$ according to the neighboring edge information of intersection $v_i$, i.e. agent $i$.
As mentioned in \chapref{section3}, each traffic intersection represents a node, and the node and edge features in the graph network can be represented by the variables in the MDP.
The traffic information collected by edge $e_j$ at time $k$ is defined the same as the value of the local observation $\left \{ o_k^i \right \} _{i=1}^N$ in the MDP
\begin{equation}
    e_j (k) = \left [ \rho _{d_j}(k),  q_{u_j, d_j}(k),  \Phi _{u_j, d_j}(k)  \right ], 
    \label{edge_expression}
\end{equation}
where $u_j$ and $d_j$ is the upstream node and downstream node of edge $e_j$, respectly.
The traffic information of edge $e_j$ can represent the interplay among agents, which offers a more comprehensive understanding of how upstream and downstream traffic flow effects propagate between adjacent intersections.
Then, the node embedding can be expressed as the aggregation of its adjacent edges traffic information
\begin{equation}
    v_{i,k}^{in}=f^{e\to v} \left ( e_j(k) \right )=f^{e\to v} \left ( e_{u_j,v_i}(k) \right ) , u_j\in \mathcal{N}_{v_i}
    \label{node_expression}
\end{equation}
where $f^{e\to v}$ is a one-layer MultiLayer Perception (MLP) with the ReLU activation function.

\subsection{Interpretable Influence Mechanism based on EHHNN} \label{Interpretable}




The traffic model is a nonlinear system, which is challenging to model for its complex spatiotemporal traffic flow data. 
Therefore, establishing an interpretable influence mechanism and extracting the impact of neighboring traffic information can enhance collaborative control among different intersections.
The EHHNN was first applied in short-term traffic flow prediction in 2022, which figured out the spatiotemporal factors influencing the traffic flow using ANOVA decomposition \cite{tao2022short}. 
Compared with this work, we further explore the application of EHHNN in large-scale multi-intersection control problem. We propose an influence mechanism based on EHHNN, which can not only capture the spatiotemporal dependencies of traffic flow data, but also illustrate the relation representation among different agents, enabling multi-agent cooperative control.
The structure of our proposed interpretable influence mechanism module is shown in the blue box in \figref{overview}.

Firstly, we reduce the dimension of the local observation $\left \{ o_k^i \right \} _{i=1}^N$ through a linear transformation layer $W$ to obtain the input variable $X_k^{in}$ for the EHH layer.
Unlike other neural networks, the EHHNN possesses interpretability, allowing us to extract the interactions among different input variables through ANOVA decomposition and an interaction matrix.
The hidden layer in EHHNN can be seen as a directed acyclic graph, as shown in the yellow box in \figref{overview}.
All nodes in the directed acyclic graph contribute to the output, including two types of neurons, source nodes $D$ and intermediate nodes $C$.

In the EHHNN, the output of source nodes can be described as:
\begin{equation}
    z_{1,s} = \max \left \{ 0,x_m-\beta _{m,q_m} \right \} 
\end{equation}
where $m$ represents the dimension of the input variable, and $\beta _{m,q_m}$ represents the $q_m$-th bias parameters on the input variable $x_m$.

In the hidden layer of the EHHNN, the intermediate nodes are obtained by minimizing existing neurons of the previous layers, which comes from different input dimension
\begin{equation}
    \begin{aligned}
        z_{p,s} = & \min _{\mathrm{nn}_{s_1}, \dots ,\mathrm{nn}_{s_p} \in J_{p,s}} \bigg \{ \max \left \{ 0,x_{\mathrm{nn}_{s_1}}-\beta _{s_1} \right \}, \\        
        & \dots \max \left \{ 0,x_{\mathrm{nn}_{s_p}} - \beta _{s_p} \right \}   \bigg \} 
    \end{aligned}
\end{equation}
where we define $J_{p,s} = \left \{ \mathrm{nn}_{s_1}, \dots, \mathrm{nn}_{s_p} \right \} $ contains the indices of neurons generated by previous layers, and $\left | J_{p,s} \right | = p $, which represents the number of interacting neurons of the $p$-th layer.

Finally, the output of the EHHNN $H_k = \left \{ h_{i,k} \right \} _{i=1}^N$ is the weighted sum of all neurons in the hidden layer:
\begin{equation}
    H_k= \alpha _0 + \sum _{s=1}^{n_1} \alpha _{1,s} z_{1,s} (\tilde{x}) + \sum _{p=2}^P \sum _{s=1}^{n_p} \alpha _{p, s} z_{p,s}(\tilde{x})
    \label{EHHNN}
\end{equation}
where $\alpha_{1,s}, \alpha_{p,s}$ are the weight of the EHHNN and $\alpha_0$ is the constant bias, $\tilde{x} = X_k^{in}$, $n_1$ and $n_p$ denotes the number of neurons in the $1$-th and $p$-th layer, respectively.

In this paper, we employ a two-factor analysis of variance (ANOVA) to determine the main effect of different traffic flow information as individual factors, as well as the interaction effect of bivariate factors on intersection congestion, i.e. $P=2$ in \equref{EHHNN}. 
The first sum represents the influence of individual variables, the second sum represents the joint influence of two variables when $P=2$.
This characteristic of the EHHNN provides insights into how different variables contribute to the overall prediction and facilitates a deeper understanding of the underlying relationships within the data.
Similar to the related work in \cite{friedman1991multivariate}, \cite{xu2020efficient} and \cite{tao2022short}, we can identify the hidden nodes that influence each output component in the ANOVA decomposition, which is calculated by
\begin{equation}
\begin{split}
    & \sigma _{m} = \sqrt{\mathrm{VAR}(f_{m}(\tilde{x}))} \\
    & f_{m}(\tilde{x}) = \sum _{J_{p,s}=\left \{ m \right \} } \alpha_{p,s} z_{p,s}(\tilde{x}) 
\end{split}
\label{anova}
\end{equation}
where $\mathrm{VAR(\cdot)}$ denotes the corresponding variance of the prediction output related to the $m$-th component of input variable $\tilde{x}$. The larger the value of $\sigma _{m}$, the greater the impact of its corresponding input variable on the degree of congestion of the traffic network.

\begin{remark}
    Due to the physical connection between traffic data and road networks, many researchers employ relational reasoning through graph-based methods or attention mechanism.
    However, the graph-based methods require a predefined and fixed adjacency matrix to reflect the spatial dependencies of different nodes, which may not effectively capture the spatiotemporal dependencies in the dynamic traffic flow data.
    The attention mechanism computes the attention coefficient by performing the dot products for all input vectors and utilizes a nonlinear activation function, making it challenging to interpret the relationships between learned weights.
    Furthermore, most traditional neural networks lack interpretability, making it challenging to select and analyze the input variables while accurately predicting the complex spatiotemporal traffic flow.
    The EHHNN can meet the requirements above, which employs a unique network structure to capture spatiotemporal dependencies from the input data.
    It has fewer parameters due to the sparsity of the network structure and can extract the influence coefficients derived from ANOVA decomposition without knowing the node connectivity of the data, as illustrated in the framework shown in \figref{ANOVA_framework}.
\end{remark}

\begin{figure}[!htp]
    \centering
    \includegraphics[width=5CM]{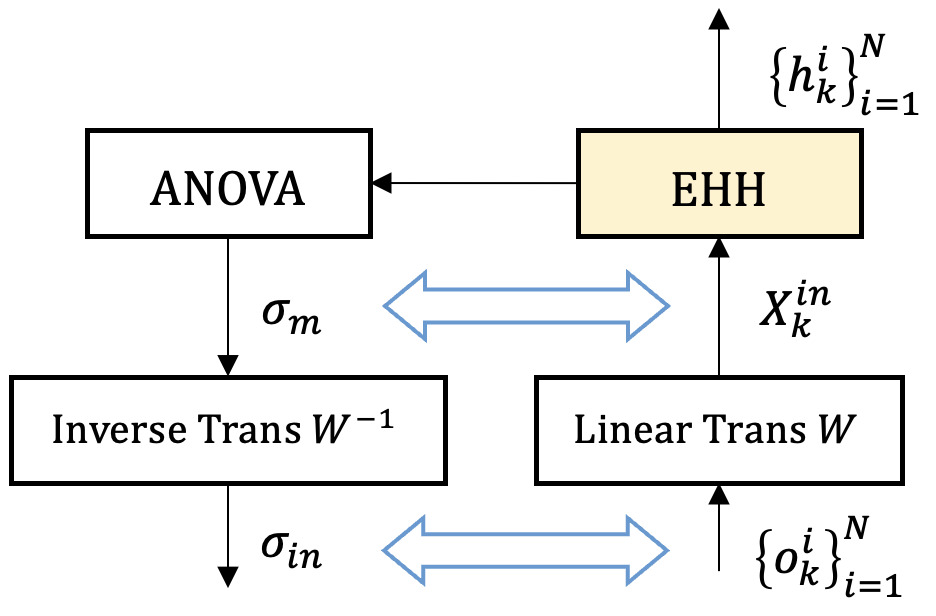}
    \caption{Structure of EHH network decomposition}
    \label{ANOVA_framework}
\end{figure}

While achieving precise prediction, it processes interpretability through ANOVA decomposition, which obtains the important coefficients of the hidden variables $\sigma_m$ as mentioned in \equref{anova}.
Additionally, since there is no nonlinear activation function after the linear transformation layer, the importance coefficients of input variables can be derived through inverse transformation:
\begin{equation}
    \sigma _{in} = W^{-1} \sigma _m
\end{equation}

Finally, we derive the aggregation embedding $V_k^{out}=\left \{ v_{i,k}^{out} \right \}_{i=1}^N $, which aggregates information from other intersections
\begin{equation}
    v_{i,k}^{out} = \sum \sigma _m v_{i,k}^{in}
    \label{node_output}
\end{equation}
where $v_{i,k}^{in}$ is the node embedding.

This output graph $V_{k}^{out}$, which extracts the spatiotemporal features through ANOVA decomposition of EHHNN, is used as the input of the centralized critic to learn a joint value function and faciliate collaboration among different agents.
While in decentralized execution, each actor learns its policy function based on the locally observed embedding vector $V_{k}^{in}$.

\subsection{Actor-Critic Framework based on BReLU neural network approximation}
The RL optimization problem is to obtain the optimal strategy $\mu ^*=\left\{u_0, u_1, \dots \right\}$ that satisfies the constraints, while maximizing the total reward, which can be written as
\begin{equation}
    \begin{aligned}
        \tilde{V}_{k+1}(x) = \max _{u\in \pi (x)} E \left \{ g(x,u,w)+\gamma \tilde{V}_k(f(x,u, w)) \right \}
    \end{aligned}
    \label{RL_opt}
\end{equation}
where $x$ denotes the state, $w$ is a random disturbance with a probability distribution $P(\cdot|x, u)$, $\tilde{V}$ is the value function, $\pi$ denotes the policy function, $g(x,u,w)$ is the cost per step, $\gamma$ is the discount factor.

\begin{lemma} \label{lemma1}
    When the value function $\tilde{V}$ is a PWL function, the policy function $\pi$ is also a PWL function.
\end{lemma}
\begin{proof}



    According to the Bellman's equation, we have
    \begin{equation}
        \begin{aligned}
            \tilde{V}^*(x) = \max _{u\in \pi^* (x)} E \left \{ g(x,u,w)+\gamma \tilde{V}^*(f(x,u, w)) \right \}
        \end{aligned}
    \end{equation}
    this equation can be view as the limit as $k \to \infty$ of \equref{RL_opt}. Then the optimal strategy can be derived as
    \begin{equation}
        \begin{aligned}
            & \mu ^* = \arg \max E \left \{ g(x,u,w)+\gamma \tilde{V}^*(f(x,u, w)) \right \} \\
            & u \in \pi(x)
        \end{aligned}
        \label{RL_opt2}
    \end{equation}
    The cost function in \equref{RL_opt2} is $g(x,u,w) + \gamma \tilde{V} (f(x,u, w))$, while $g(x,u,w)$ is aslo known as the reward function in RL.
    In this paper, the reward function as shown in \equref{reward2} is a PWL function with respect to state $x$.
    If the value function $\tilde{V}$ is a PWL function, the cost function is the sum of two PWL functions, which is also a PWL function.
    It has been proved by bemporad in 2002 \cite{bemporad2002explicit} that minimizing or maximizing a PWL cost function $\tilde{V}$ over a polyhedron yields a PWL solution $\pi$.
    Therefore, if we approximate the value function $\tilde{V}$ using a PWLNN, the policy function $\pi$ should also be a PWL solution, as stated in the conclusion above.
\end{proof}


To align with the conclusion obtained from Lemma \ref{lemma1}, we employ PWL neural networks to approximate the value function and policy function in actor-critic.
The activation function plays an important role in neural network approximation, in which the ReLU activation function is prevalent in neural networks due to its simplicity and computation efficiency, and it can be defined as:
\begin{equation}
    z(x)=\mathrm{max}\left \{ 0,x\right \}
\end{equation}
where $x=[x_1,x_2,\dots,x_n]^T \in \mathbb{R}$.
The prevalent ReLU network is a kind of HH neural network \cite{liang2021biased}.

Based on ReLU neural network, in 2021, Liang and Xu \cite{liang2021biased} proposed BReLU neural network, which can be expressed as:
\begin{equation}
    z(x)=\mathrm{max}\left \{ 0,x_i-\beta_{i_1}  \right \} ,\dots,\mathrm{max}\left \{ 0,x_i-\beta_{i_{q_i}}  \right \}
    \label{neurons}
\end{equation}
where $q$ represents the number of linear regions in each dimension of the input data.

Different from ReLU neuron, BReLU uses different bias parameters $\beta_{i_{q_i}}$ for variables in various dimensions, thereby partitioning the input spaces into a number of linear regions and take advantage of the characteristics of PWL functions.
The input variables after normalization follow an approximately normal distribution, and the multiple bias parameters $\beta_{i_{q_i}}$ are determined by the distribution of the input data:
\begin{equation}
    \beta _{i_{q_i}}=[-3\nu  ,-0.824\nu  ,-0.248\nu  ,0.248\nu  ,0.834\nu ]+\eta
\end{equation}
where the parameter $\nu , \eta$ represent the variance and expectation of input variable after normalization, respectly. $q_i$ is the number of bias parameters of the $i$-th layer and also represents the number of linear sub-regions.
The selection of the value $q_i$ is a trade-off between network accuracy and the number of parameters.
And the weights in BReLU neural network are obtained using the backpropagation method.

The BReLU neural network demonstrates higher flexibility and better approximation capabilities by effectively dividing the input spaces into numbers of linear regions, even when the output increases exponentially with the input.
Therefore, we approximate the value function $\tilde{V}$ and policy function $\pi$ in the actor-critic framework using BReLU neural network, where the functions approximated by BReLU neural network are PWL.
The reason for doing this is two-fold.
First, it coincides with the conclusion that minimizing (maximizing) PWL functions over polyhedron yields PWL solutions.
Second, the approximation of the value function and policy functions using BReLU neural network provides a more precise approximation than that of ReLU.
The red box in \figref{overview} shows the structure of the BReLU neural network used in the proposed PWL-actor-critic framework.

Different from independent proximal policy optimization (IPPO)\cite{guo2020joint}, which lacks collaboration among agents, in our proposed multi-agent actor-critic algorithm, all critics are updated together to minimize a joint regression loss function as shown in \equref{critic}, where the joint value function remains a PWL function. And $\hat{A}_i$ represents the advantage function, which is estimated by the truncated version of the generalized advantage estimation \cite{schulman2015high}, $N_{b}$ denotes the training batch size, $b'$ is an arbitrary sampling sequence after sample $b$, $\gamma$ is the discount factor, $\lambda $ is the regularization parameters and $W$ is the weight of the neural network that approximates the value function $\tilde{V}$.

\begin{equation}
    \begin{aligned}
        L(\phi) & = - \sum_{i=1}^N \sum_{b=1}^{N_{b}} (\hat{A}_i(b))^2 + \lambda  \left \| W \right \| _1\\
        \hat{A}_i(b) & =\sum_{b'>b} \gamma ^{b'-b}r_i(b')-\tilde{V}(v_{i,b}^{out})
    \end{aligned}
    \label{critic}
\end{equation}

Each decentralized actor updates its policy function only based on its local observations, which is updated using the proximal policy optimization method. And a clip function is used to limit the range of changes in the probability ratio $r_t(\theta,b)$ of old and new strategies $\pi _{i, \theta ^{-}}, \pi _{i, \theta }$ to avoid large variance changes and unstable training.
\begin{equation}
    \begin{aligned}
        L_i(\theta) & = \sum_{b=1}^{N_{b}} \min \left ( r_t(\theta, b), \mathrm{clip} \left ( r_t(\theta,b), 1-\epsilon, 1+\epsilon \right ) \right ) \cdot \hat{A}_i(b)\\
        r_t(\theta,b) & = \frac{\pi_{i, \theta}(a_b|v_{i,b}^{in})}{\pi_{i, \theta ^-}(a_b|v_{i,b}^{in})} 
    \end{aligned}
\label{actor}
\end{equation}
where

To give an overview of the proposed multi-agent algorithm, we now summarize it in \algoref{BRGEHH} and briefly introduce the training process.
During a total of $N_{ep}$ episodes, each lasting for a duration of $T$ in the training process, we store the transition $\left\{(o_i(k), a_i(k), o_i(k+1), r_i(k))\right\}_{i=1}^N$ into the memory buffer $\mathcal{M}$.
When the recorded data exceeds the batch size $N_b$, we use the pre-trained EHHNN to compute the importance coefficients of the input features and update both the actor and critic network with constant learning rate.

\begin{algorithm}[h] \small
	\renewcommand{\algorithmicrequire}{\textbf{Input:}}
	\renewcommand{\algorithmicensure}{\textbf{Output:}}
	\caption{Multi-agent BReLU actor-critic framework} 
    \label{BRGEHH}
	\begin{algorithmic}[1] 
		\Require
        A pre-trained EHHNN
		\Ensure 
		actor $\theta_i$ for $i \in N$, critic $\phi$
		\State Initialization: actor $\theta _i \gets 0$ and critic $\phi _i \gets 0$ for $i \in N$;
            \For{$ep=1,\dots,N_{ep}$}
                \For{$k=1,\dots,T/\Delta t$}
                \State Reset the environment, $o_i(k), \mathcal{M}=\emptyset $
                \State Sample $a_i(k)$ from $\pi_i(k)$
                \State  receive $r_i(k)$ (\ref{reward2}) and $o_i(k+1)$
            \EndFor
            \State Store transitions $\mathcal{M} \gets \mathcal{M} \cup [o_i(k),a_i(k),r_i(k)]_{i=1}^N$
            \If{$N(\mathcal{M})>N_b$}
                \State Compute $\sigma$ using the pre-trained EHHNN (\ref{anova})
                \State Compute the aggregated result $v_{i,k}^{out}$ using (\ref{node_output})
                \For{$i=1,\dots, N$}
                \State Calculate the advantage function $\hat{A}_i(b)$ (\ref{actor})
                \State Update the actor $\theta_i$ by minimizing $L_i(\theta)$ (\ref{actor})
                \EndFor
                \State Update the global critic $\phi$ by minimizing $L(\phi)$ (\ref{critic})
            \EndIf
            \EndFor
	\end{algorithmic} 
\end{algorithm}

\section{Simulation Results} \label{chapter5}
In this section, we evaluate our proposed method using the SUMO traffic simulator \cite{behrisch2011sumo}.
Firstly, we employ
traffic data of Los Angeles country (METR-LA) dataset \cite{li2017diffusion} for traffic forecasting, comparing existing attention mechanism and analyzing the effectiveness of the proposed interpretable influence mechanism based on EHHNN.
Compared with the work in \cite{tao2022short}, we conduct the result on a different dataset and validate the performance of EHHNN on a much larger traffic network.
Subsequently, we conduct both quantitative and qualitative experiments for multi-intersection traffic signal control on two synthetic traffic networks, and compare it to the traditional fixed time control method and the SOTA MARL controllers.

\subsection{Traffic forecasting}

\subsubsection{Dataset Description}
To validate the effectiveness and interpretability of the proposed influence mechanism, we conduct experiments on the METR-LA dataset
.
The METR-LA dataset consists of traffic data collected from circular detectors on highways in Los Angeles Country, comprising a total of 207 sensors.
For our experiments, we selecte 15 adjacent nodes, as illustrated in \figref{METR}.
We evaluate the forecasting performance of the EHHNN compared with four baseline neural networks.
Furthermore, we conducte the interpretability analysis of the EHH network to facilitate a deeper understanding of the underlying relationships between different nodes.

\begin{figure}[!htp]
    \centering
    \includegraphics[width=7cm]{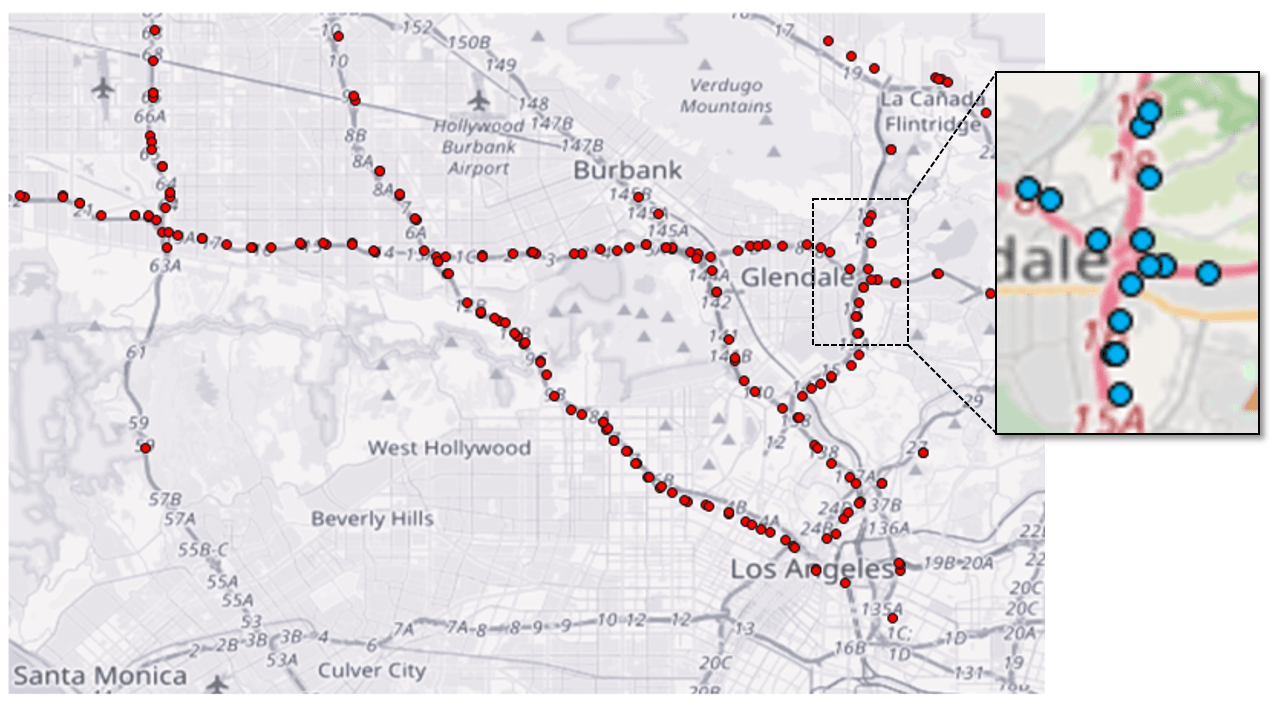}
    \caption{The METR-LA traffic network}
    \label{METR}
\end{figure}

For the traffic data in METR-LA dataset, information is collected every 5 minutes, with an observation window of 60 minutes and a maximum prediction horizon of 45 minutes.
We split the dataset into three distinct sets, including a training set, a validation set and a test set, with the training set accounting for 60\% of the total samples, and the remaining 40\% each allocated to the validation and test sets.

\subsubsection{Measures of effectiveness}
To evaluate the accuracy of different forecasting models, we employ three measures of effectiveness,
including mean absolute error (MAE), $R^2$ value\cite{niri2021machine} and root mean square error (RMSE):
\begin{equation*}
    \mathrm{MAE}=\frac{1}{N\cdot T} \sum_{i=1}^N \sum_{t=1}^T \left | (x_{i,t}-\hat{x}_{i,t} ) \right | 
\end{equation*}
\begin{equation*}
    \mathrm{R}^2=1-\sum _{i=1}^N \sum _{t=1}^T (x_{i,t}-\hat{x}_{i,t})^2 / \sum _{i=1}^N \sum _{t=1}^T (x_{i,t}-\bar{x} )^2
\end{equation*}
\begin{equation}
    \mathrm{RMSE} = \sqrt{\frac{1}{N\cdot T}\sum _{i=1}^N \sum_{t=1}^T(x_{i,t}-\hat{x}_{i,t} )^2} 
\end{equation}
where $N$ denotes the number of predicted nodes of METR-LA dataset, $T$ denotes the historical time window for prediction, $x_{i,t}$ is the input traffic flow data of the $i$-th node at time $t$, and $\bar{x}$ represents the mean value of the input traffic flow data $x_{i,t}$.





\subsubsection{Experimental Settings}
All experiments are compiled and tested on Linux cluster (CPU: Intel(R) Xeon(R) Silver 4216 CPU @ 2.10GHz GPU:NVIDIA GeForce RTX 3090).
The EHHNN is a single hidden layer neural network.
In the traffic forecasting experiments of this section, our model initially reduces the input dimension through a linear transformation layer.
Then, the hidden variables after dimension reduction are fed into the EHHNN, which outputs the prediction results.
Our proposed model is effectively a two-layer shallow neural network.
Therefore, we compare it with three other shallow networks, including a two-layer fully connected (FC) neural network, a FC long short-term memory (FC-LSTM) neural network, and graph attention neural network (GAT).
Besides, we compared the EHHNN with a SOTA method in traffic forecasting, which is a deep network called Spatial-Temporal Graph Convolutional Network (STGCN).
All models utilized the same training parameters, undergoing 300 epochs of training on the training set, employing a variable learning rate optimizer.
Subsequently, the optimal parameters of the model were determined through performance evaluation on the validation set.
Finally, the models were evaluated on the test set.
All tests were conducted with a historical time window of 60 minutes, i.e. 12 sampling points, for predicting the traffic condition in the subsequent 15, 30, and 45 minutes.

\subsubsection{Experiment Results}
Table \ref{METR-result} demonstrates the results of EHHNN and baselines on the dataset METR-LA.
Compared with the three shallow networks, the EHHNN achieved the best performance.
In comparison with the STGCN deep network, our model outperformed in both $\mathrm{RMSE}$ and $\mathrm{R}^2$ metrics, indicating better fitting to large errors and overall predictions that better align with the actual trends in values.

\begin{table*}[!htp]
    \centering
    \caption{Comparison of performance on METR-LA}
    \label{METR-result}
    \label{Comparison of the performance on METR-LA dataset}
    \begin{tabular}{c|c|c|c|c|c}
        \hline
        \multirow{2}{*}{Model}
            & \multicolumn{5}{c}{(15/30/60 min)}\\
        \cline{2-6}
        & MAE $\downarrow$ & R$^2$ $\downarrow$ & RMSE $\uparrow$ &  Params & Neurons\\
        \hline
        2-layers FC & 8.26949/8.42299/8.70810 & 0.08866/0.07988/0.06103 & 19.86146/19.96055/20.17158 & 26029/28954/34804 & 64\\ 
        FC-LSTM & 4.01271/4.63315/5.66560 & 0.77434/0.72570/0.64237 & 9.78062/10.85584/12.69706 & 59181/93786/162996 & 64\\
        GAT & 5.30224/6.01932/6.77592 & 0.79210/0.71892/0.63816 & 9.48623/11.03227/12.52191 & 14339/15878/18956 & 64\\
        STGCN & \textbf{2.83072}/ \textbf{3.61224}/\textbf{4.60460} & 0.83381/0.74521/0.64770 & 8.48147/10.50363/12.35573 & 96767/97154/97928 & 960\\
        EHHNN & 3.35360/4.21736/5.24809 & \textbf{0.84352}/\textbf{0.77446}/\textbf{0.66847} & \textbf{8.22992}/\textbf{9.88232}/\textbf{11.98600} & 42763/62293/101353 & 370 \\
    \hline    
    \end{tabular}
\end{table*}

Besides, due to the sparse connectivity and the neurons-connected structure of the EHHNN, it can achieve superior prediction accuracy with fewer training parameters.
The number of hidden layer neurons significantly impacts the performance and capability of a neural network.
Increasing the number of hidden layer neurons enhances the expressive ability and complexity of the neural network.
However, in a FC network, adding neurons implies increasing training parameters, leading to longer training times and heightened demands on computational resources.
Compared with the FC network, FC-LSTM network, the EHHNN requires fewer training parameters when having the same number of hidden neurons.
\subsection{Multi-intersection traffic signal control}

After evaluating the effectiveness and interpretability of the EHHNN in traffic forecasting, in this section, we apply the proposed influence mechanism based on EHHNN to the multi-intersection traffic signal control problem, and evaluate the algorithm on two different synthetic traffic networks.

\subsubsection{Traffic Signal Control using Synthetic Traffic Networks}
We evaluate our proposed method on two synthetic traffic networks, including a 5$\times$5 traffic grid and a non-Euclidean traffic network, as shown in \figref{Synthetic_traffic_network}.
The details are introduced as follows:
\begin{itemize}
    \item Network$_{5\times 5}$: A $5\times 5$ traffic grid with three bidirectional lanes in four directions at each intersection. The road length is 100m, and the lane speed limit is 13.89m/s. There are approximately 930 vehicles generated and added to the network per episode.
    \item Network$_{NonE}$: An non-Euclidean traffic network with 8 intersections. The intersections are not connected in a grid-like pattern, and the length of roads varies between 75m to 150m. There are 10 external road inputs, and approximately 250 vehicles are generated and added to the network per episode. We designed this network for the various buliding areas of the cities, as described in Subsection \ref*{modeling}.
\end{itemize}

\begin{figure}[!htp]
    \centering
    \subfigure[5$\times$5 traffic grid]{
        \includegraphics[width=4cm]{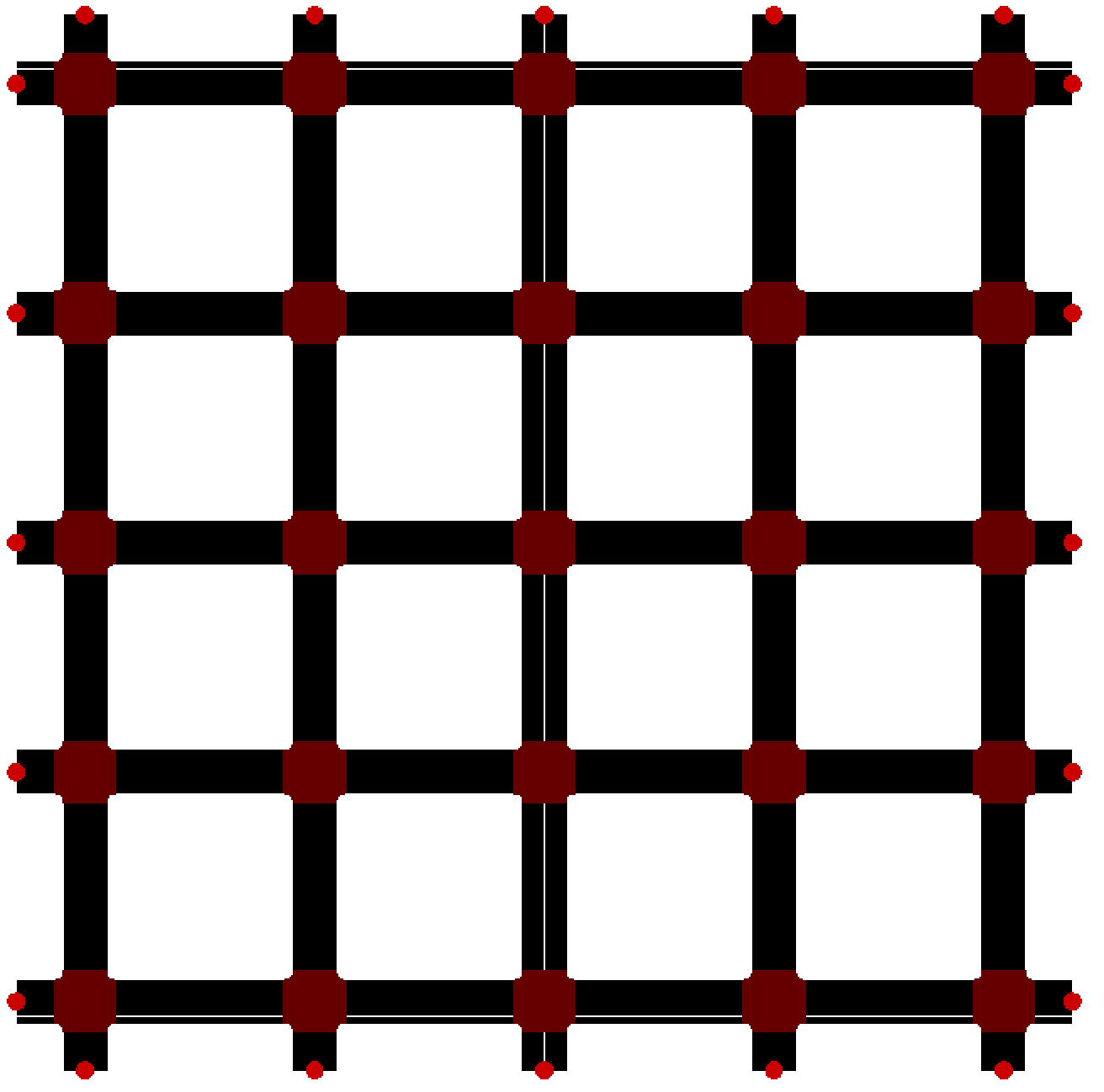}
        \label{grid_network}
    }
    \subfigure[Non-Euclidean traffic network]{
        \includegraphics[width=4cm]{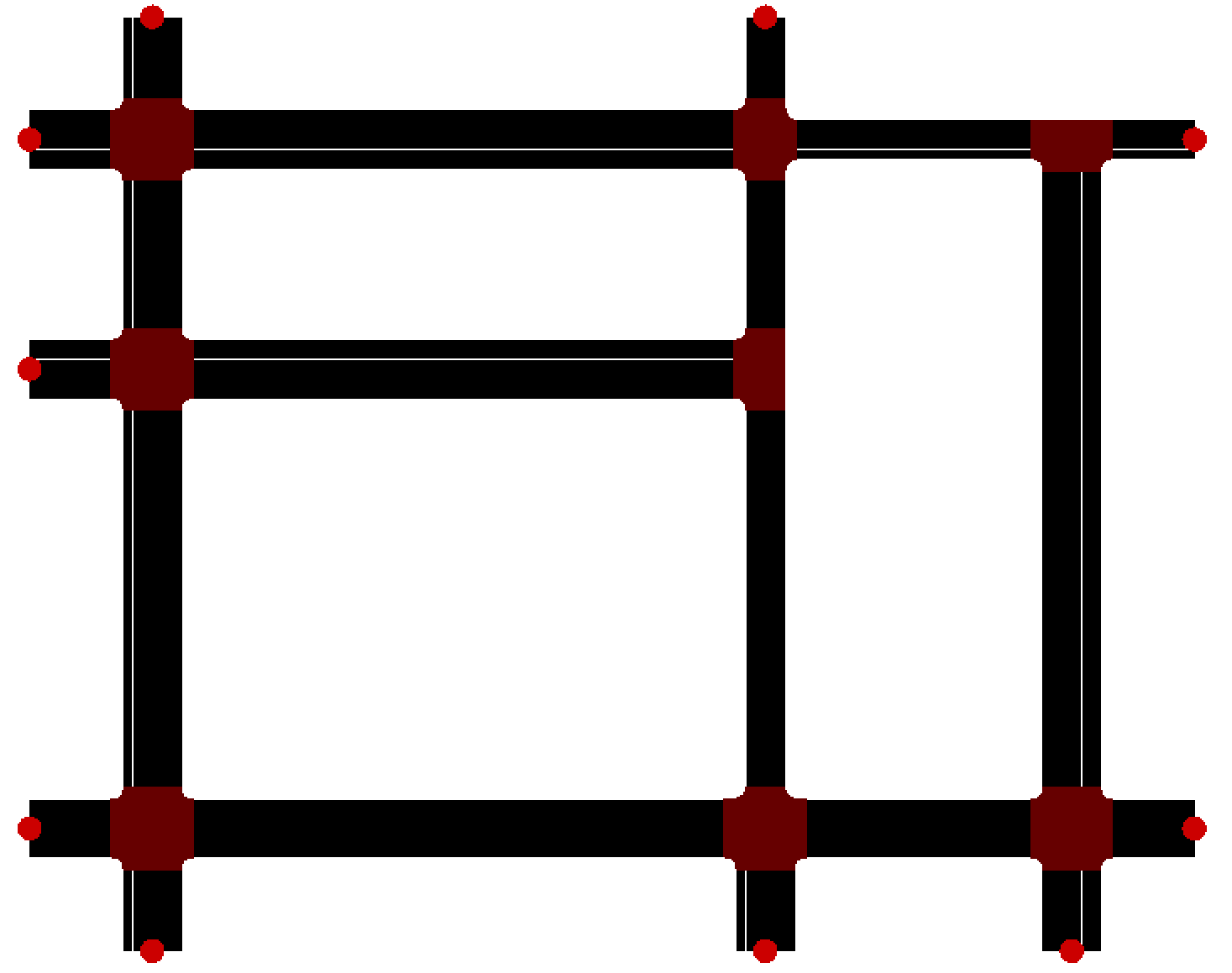}
    }
    \caption{Synthetic Traffic Networks}
    \label{Synthetic_traffic_network}
\end{figure}

\subsubsection{Evaluation Metric}
To quantitatively measure the condition of the traffic network, we define two metrics, as shown in \equref{ave_waiting} and \equref{stability}. The average waiting time (AVE) is used to measure the overall congestion degree of the traffic network, while the traffic flow stability (STA) indicates the frequency of local vehicle accumulation that appears during the simulation time.
\begin{equation}
    \mathrm{AVE} = \frac{1}{T_s} \sum_{t=0}^{T_s} \sum_{i=1}^N W_i(t)
    \label{ave_waiting}
\end{equation}
\begin{equation}
    \mathrm{STA} = \frac{1}{T_s}\sum_{t=0}^{T_s}\left [ \sum_{i=1}^N W_i(t)-E \right ] ^2 
    \label{stability}
\end{equation}
where $T_s$ denotes the simulation time, $N$ denotes the number of intersections, i.e., agents.
\subsubsection{Methods Compared with}
To evaluate the efffectiveness of our proposed algorithm,
we compare it with serval SOTA traffic signal control methods,
including both tradictional approach and RL methods.

\begin{itemize}
    \item Fixed time control: a pre-defined rule-based traffic control method, which uses a four-stage phase sequence control.
    \item Independent Deep Q Network (IDQN): a decentralized method where each agent learns a Q network to maximize its reward independently without interacting with each other. 
    \item Multi-Agent Deep deterministic Policy Gradient (MADDPG): each agent learns a deterministic policy, and this method takes into account the interaction between different agents by sharing the experience pool and employing collaborative training methods.
    \item Independent Proximal Policy Optimization (IPPO): a decentralized actor-critic-based method where each agent learns a truncated advantage function and a policy function to improve performance.
    \item Graph Convolution RL (DGN): a standard graph-based RL method that uses graph convolution neural network and self-attention mechanism to extract the traffic feature, thereby learning the Q-values of the agents.
 \end{itemize}

\subsubsection{Simulation Settings}
The simulation runs $T=2500s$ per episode,
and the traffic signal controllers update every $\Delta t=5s$.
To ensure the safety at the intersection,
we set the yellow phase to 2s,
the minimum and maximum green time $g_{min}, g_{max}$ to 5s and 50s, respectively.

In the fixed time control, we set the green phase for each traffic signal controller as 25s.
And for all RL methods, we use the same training parameters.
All methods are trained for 100 episodes, Table \ref{MDPset} shows the hyperparameters settings of the reward function in MDP and \algoref{BRGEHH}.

\begin{table}[!htp]
    \begin{center}
        \caption{Hyperparameter settings in MDP}
        \setlength\tabcolsep{4pt}
        \begin{tabular}{ccc}
            \hline
            Variable & Notation & Value \\ 
            \hline
            extra reward & $\kappa_1$ & 25 \\
            reduce congestion parameter & $\kappa _2$ & 5 \\
            increase congestion parameter & $\kappa _3$ & 5 \\
            Q network learning rata & $\alpha _Q$ & 0.01 \\
            critic learning rate & $\iota _{\mathrm{c}}$ & 0.01 \\
            actor learning rate & $\iota _{\mathrm{A}}$ & 0.001 \\
            hyperparameter of clip function & $\epsilon$ & 0.2 \\
            minibatch size & $N_{b}$ & 32 \\ \hline
        \end{tabular}
        \label{MDPset}
    \end{center}
\end{table}

\subsubsection{Experimental Results}
We compare our proposed method with the baseline methods on both the 5$\times$5 traffic grid and non-Euclidean traffic network.
\figref{grid5x5_1} shows the comparison results of our proposed method and all the baseline algorithms on a 5$\times$5 traffic grid.
It is evident that, under the current parameter settings, the IDQN and MADDPG methods fail to alleviate traffic congestion, as shown in \figref{grid_rw} that the reward of IDQN and MADDPG have been negative per episode.
After excluding these two algorithms, as shown in \figref{grid5x5_1}, 
our proposed algorithm achieves similar performance to the two SOTA algorithm IPPO, DGN.
Besides, compared to fixed time control, IDQN, MADDPG, our proposed mehtod achieves reduce the congestion on the traffic network.

\begin{figure}
    \centering
    \subfigure[Waiting Time Comparison]{
        \includegraphics[width=8cm]{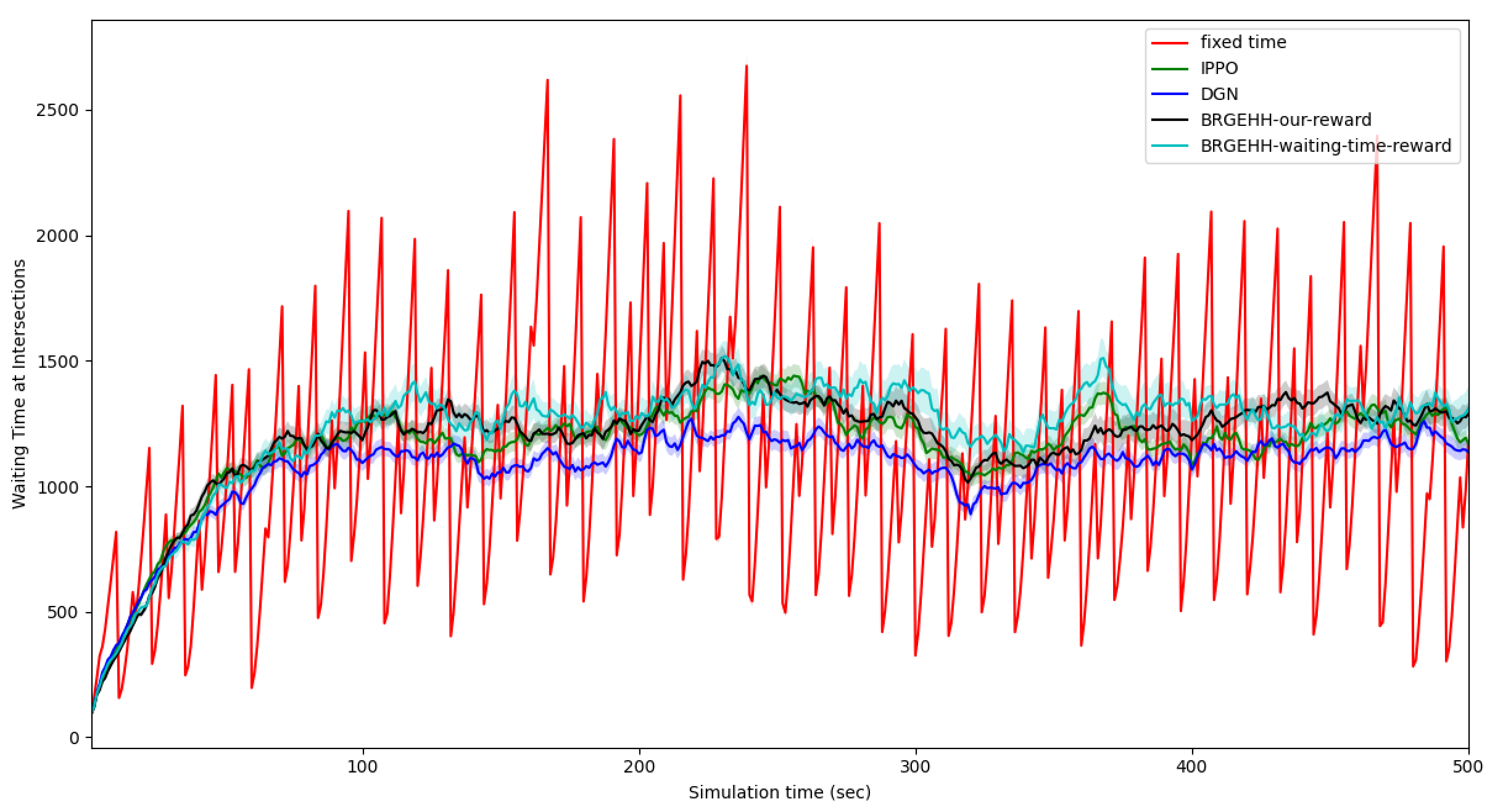}
    }
    \subfigure[Reward Comparison]{
        \includegraphics[width=4cm]{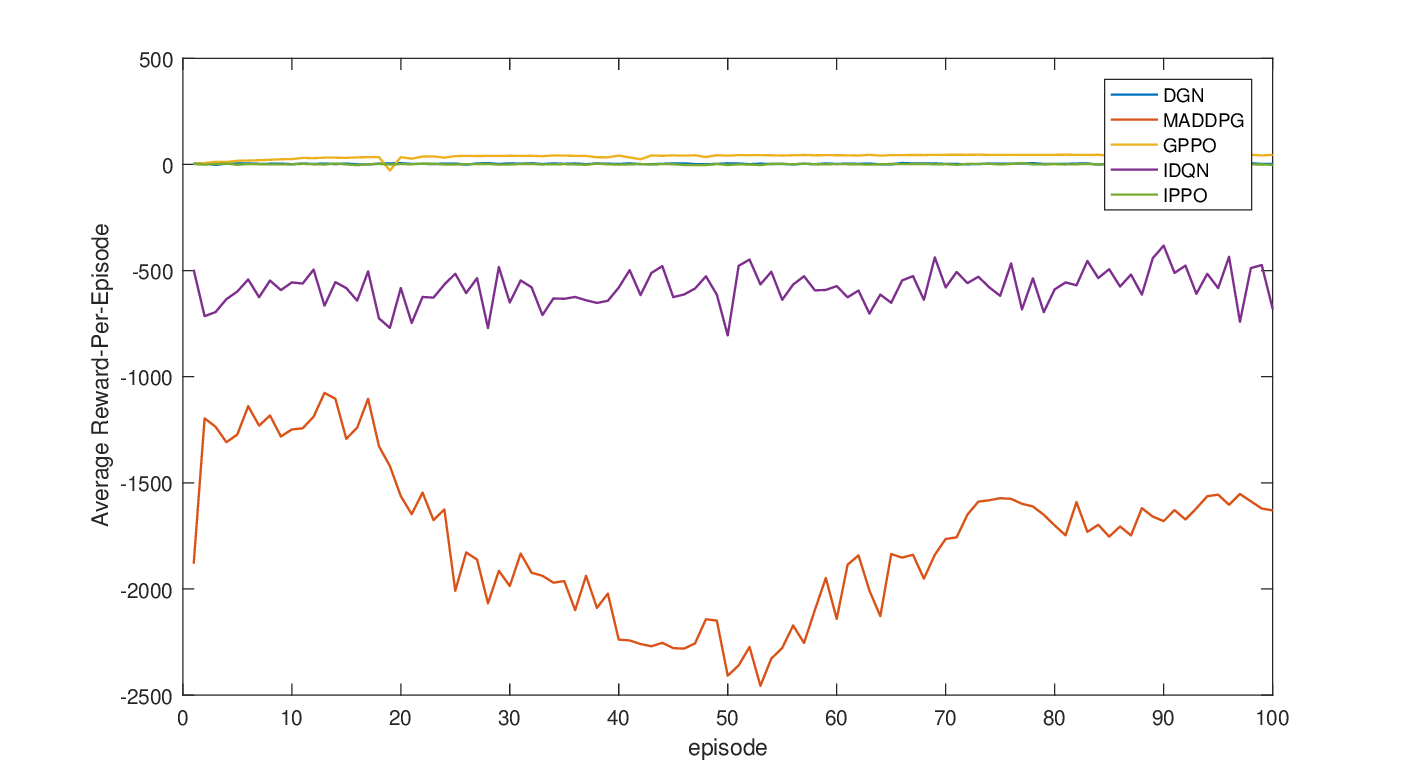}
        \label{grid_rw}
    }
    \subfigure[Flow Stability Comparison]{
        \includegraphics[width=4cm]{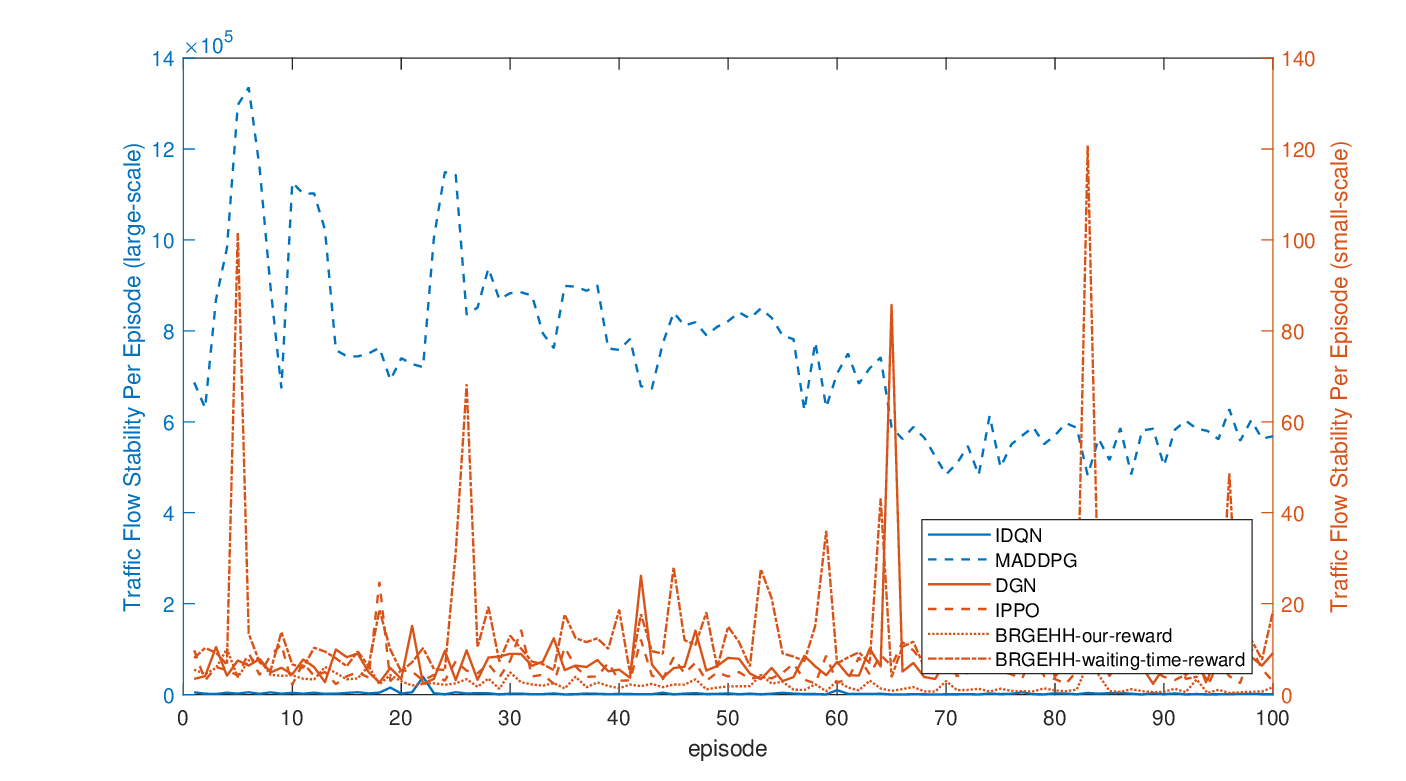}
    }
    \caption{Performance Comparison on the 5$\times$5 Synthetic Traffic Grid}
    \label{grid5x5_1}
\end{figure}

Subsequently, as illustrated in \figref{NonEudian_1}, we conduct the results on an non-Euclidean traffic network.
Our proposed algorithm exhibits better performance than IPPO, DGN and fixed time control method,
and achieves reduced delay and smoother traffic flow conditions compared to other methods.
Meanwhile, we conduct an ablation experiment, revealing that the performance obtained using reward function \equref{reward1} is inferior to that achieved using the reward function proposed in this paper.



\begin{figure}
    \centering
    \subfigure[Waiting Time Comparison]{
        \includegraphics[width=8cm]{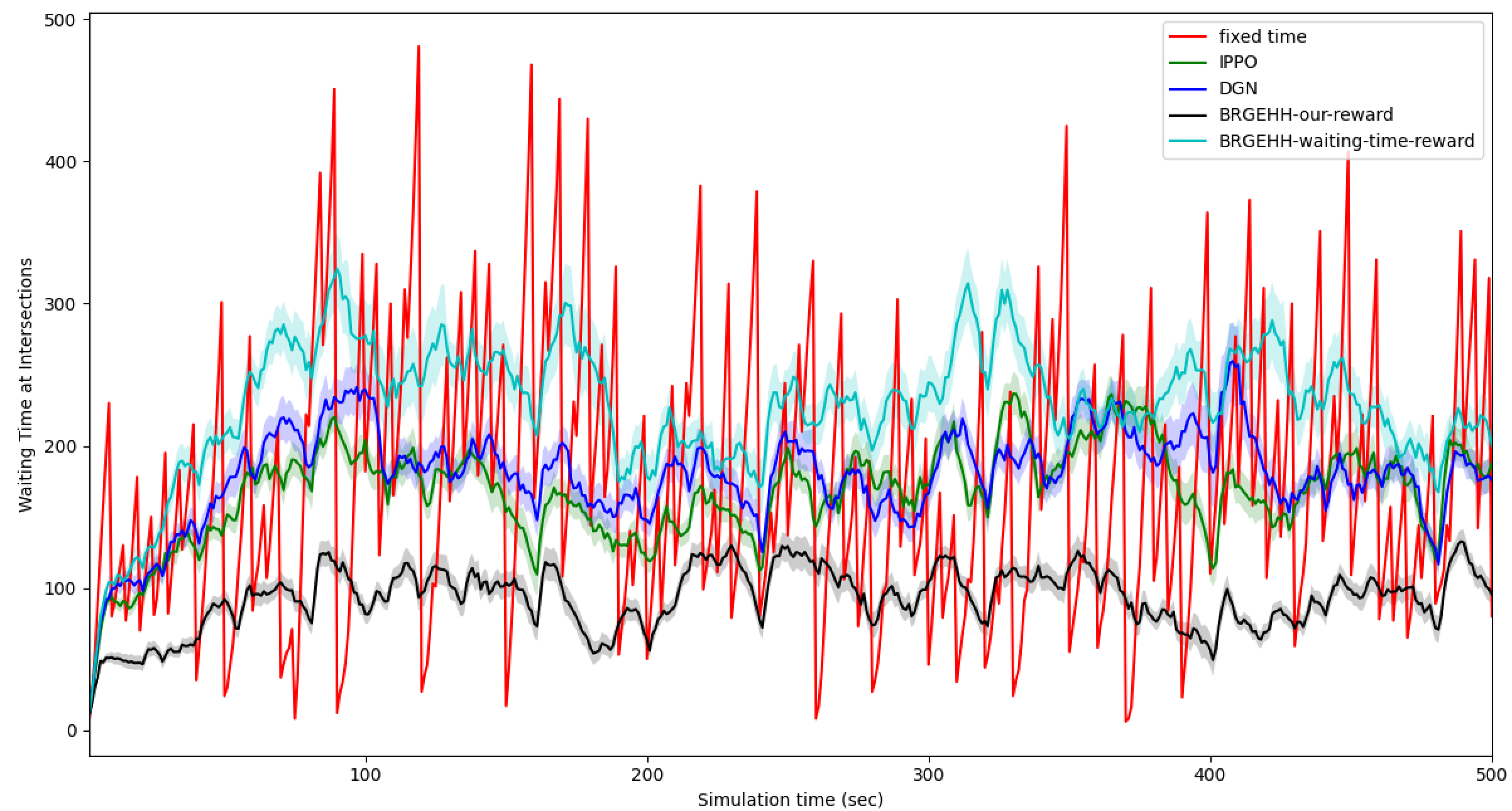}
    }
    \subfigure[Average Reward Comparison]{
        \includegraphics[width=4cm]{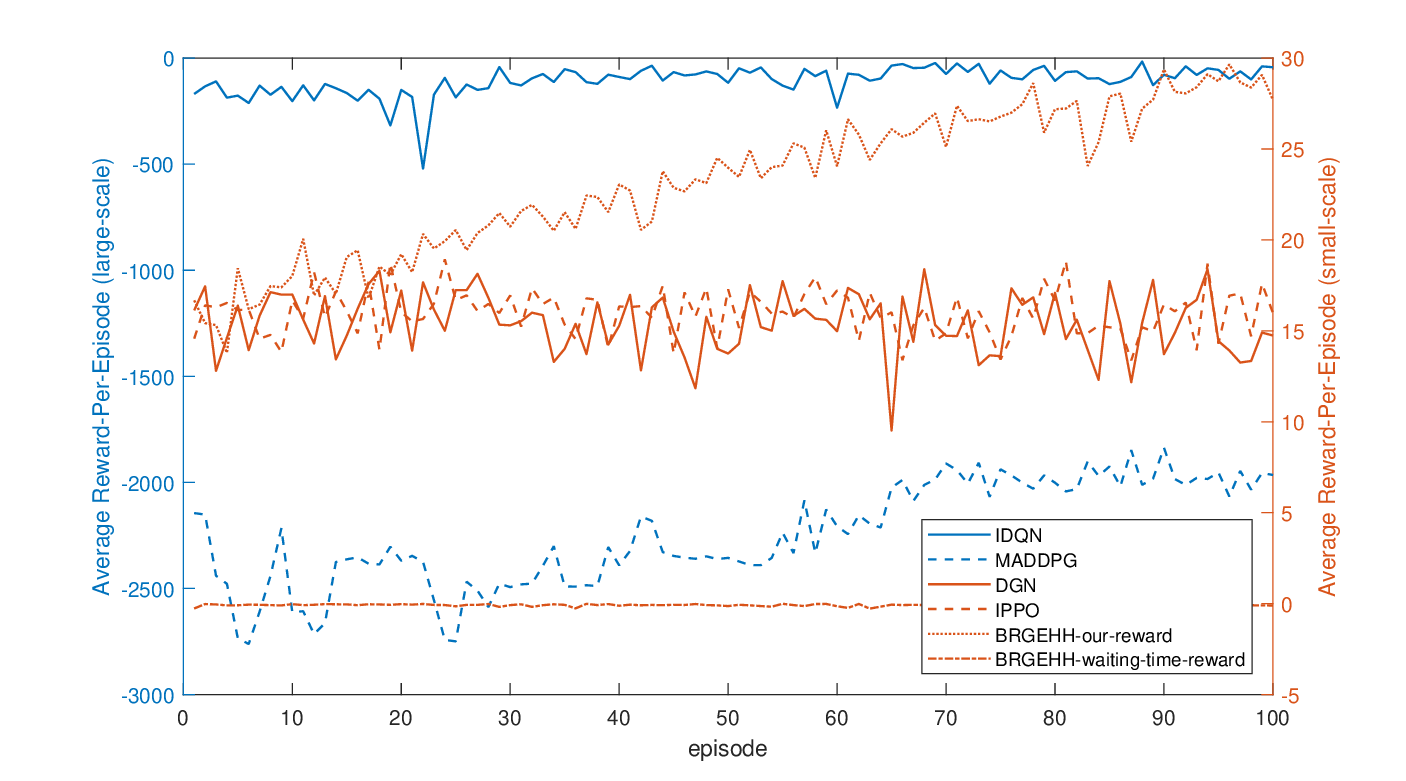}
        \label{non_rw}
    }
    \subfigure[Flow Stability Comparison]{
        \includegraphics[width=4cm]{img/non_result_sta.eps}
    }
    \caption{Performance Comparison on the Non-Euclidean Synthetic Traffic Network}
    \label{NonEudian_1}
\end{figure}

\section{Conclusion} \label{chapter6}

In this paper, we have introduced a novel multi-agent actor-critic framework with an interpretable influence mechanism based on the EHHNN.
Specifically, we used the BReLU neural network as a function approximator for both the value and policy functions and thus construct the PWL-actor-critic framework.
Besides, the proposed influence mechanism based on the EHHNN can capture the spatiotemporal dependencies in traffic information without knowing the pre-defined adjacency matrix. The importance of the input variables using ANOVA decomposition of EHHNN, providing a deeper understanding of the underlying relationships within the data.
Moreover, the approximation of the global value function and local policy functions using BReLU neural network not only provides a more precise approximation but also coincides with the conclusion that minimizing PWL functions over polyhedron yields PWL solutions. 
Simulation experiments on both the synthetic traffic grid and non-Euclidean traffic network demonstrate the effectiveness of the proposed multi-agent actor-critic framework, which can effectively extract important information and coordinate signal control across different intersections, resulting in lower delays in the whole traffic network.
\ifCLASSOPTIONcaptionsoff
  \newpage
\fi

\begin{refcontext}[sorting=none]
\printbibliography

@article{gupta2018low,
  title={A low-cost open hardware system for collecting traffic data using Wi-Fi signal strength},
  author={Gupta, Shivam and Hamzin, Albert and Degbelo, Auriol},
  journal={Sensors},
  volume={18},
  number={11},
  pages={3623},
  year={2018},
  publisher={MDPI}
}

@article{zhao2011computational,
  title={Computational intelligence in urban traffic signal control: A survey},
  author={Zhao, Dongbin and Dai, Yujie and Zhang, Zhen},
  journal={IEEE Transactions on Systems, Man, and Cybernetics, Part C (Applications and Reviews)},
  volume={42},
  number={4},
  pages={485--494},
  year={2011},
  publisher={IEEE}
}

@article{xu2020efficient,
  title={Efficient hinging hyperplanes neural network and its application in nonlinear system identification},
  author={Xu, Jun and Tao, Qinghua and Li, Zhen and Xi, Xiangming and Suykens, Johan AK and Wang, Shuning},
  journal={Automatica},
  volume={116},
  pages={108906},
  year={2020},
  publisher={Elsevier}
}

@article{friedman1991multivariate,
  title={Multivariate adaptive regression splines},
  author={Friedman, Jerome H},
  journal={The Annals of Statistics},
  volume={19},
  number={1},
  pages={1--67},
  year={1991},
  publisher={Institute of Mathematical Statistics}
}

@article{tao2022short,
  title={Short-term traffic flow prediction based on the efficient hinging hyperplanes neural network},
  author={Tao, Qinghua and Li, Zhen and Xu, Jun and Lin, Shu and De Schutter, Bart and Suykens, Johan AK},
  journal={IEEE Transactions on Intelligent Transportation Systems},
  volume={23},
  number={9},
  pages={15616--15628},
  year={2022},
  publisher={IEEE}
}

@article{bemporad2002explicit,
  title={The explicit linear quadratic regulator for constrained systems},
  author={Bemporad, Alberto and Morari, Manfred and Dua, Vivek and Pistikopoulos, Efstratios N},
  journal={Automatica},
  volume={38},
  number={1},
  pages={3--20},
  year={2002},
  publisher={Elsevier}
}

@article{breiman1993hinging,
  title={Hinging hyperplanes for regression, classification, and function approximation},
  author={Breiman, Leo},
  journal={IEEE Transactions on Information Theory},
  volume={39},
  number={3},
  pages={999--1013},
  year={1993},
  publisher={IEEE}
}

@article{liang2021biased,
  title={Biased ReLU neural networks},
  author={Liang, XingLong and Xu, Jun},
  journal={Neurocomputing},
  volume={423},
  pages={71--79},
  year={2021},
  publisher={Elsevier}
}

@article{guo2020joint,
  title={Joint optimization of handover control and power allocation based on multi-agent deep reinforcement learning},
  author={Guo, Delin and Tang, Lan and Zhang, Xinggan and Liang, Ying-Chang},
  journal={IEEE Transactions on Vehicular Technology},
  volume={69},
  number={11},
  pages={13124--13138},
  year={2020},
  publisher={IEEE}
}

@article{battaglia2018relational,
  title={Relational inductive biases, deep learning, and graph networks},
  author={Battaglia, Peter W and Hamrick, Jessica B and Bapst, Victor and Sanchez-Gonzalez, Alvaro and Zambaldi, Vinicius and Malinowski, Mateusz and Tacchetti, Andrea and Raposo, David and Santoro, Adam and Faulkner, Ryan and others},
  journal={arXiv preprint arXiv:1806.01261},
  year={2018}
}

@article{van2016coordinated,
  title={Coordinated deep reinforcement learners for traffic light control},
  author={Van der Pol, Elise and Oliehoek, Frans A},
  journal={Proceedings of Learning, Inference and Control of Multi-agent Systems (NIPS)},
  volume={8},
  year={2016},
  pages={21--38}
}

@inproceedings{calvo2018heterogeneous,
  title={Heterogeneous multi-agent deep reinforcement learning for traffic lights control},
  author={Calvo, Jeancarlo Arguello and Dusparic, Ivana},
  booktitle={AICS},
  pages={2--13},
  year={2018}
}

@article{casas2017deep,
  title={Deep deterministic policy gradient for urban traffic light control},
  author={Casas, Noe},
  journal={arXiv preprint arXiv:1703.09035},
  year={2017}
}

@article{chu2019multi,
  title={Multi-agent deep reinforcement learning for large-scale traffic signal control},
  author={Chu, Tianshu and Wang, Jie and Codec{\`a}, Lara and Li, Zhaojian},
  journal={IEEE Transactions on Intelligent Transportation Systems},
  volume={21},
  number={3},
  pages={1086--1095},
  year={2019},
  publisher={IEEE}
}

@article{tan2019cooperative,
  title={Cooperative deep reinforcement learning for large-scale traffic grid signal control},
  author={Tan, Tian and Bao, Feng and Deng, Yue and Jin, Alex and Dai, Qionghai and Wang, Jie},
  journal={IEEE Transactions on Cybernetics},
  volume={50},
  number={6},
  pages={2687--2700},
  year={2019},
  publisher={IEEE}
}

@article{kipf2016semi,
  title={Semi-supervised classification with graph convolutional networks},
  author={Kipf, Thomas N and Welling, Max},
  journal={arXiv preprint arXiv:1609.02907},
  year={2016}
}

@article{hamilton2017inductive,
  title={Inductive representation learning on large graphs},
  author={Hamilton, Will and Ying, Zhitao and Leskovec, Jure},
  journal={Advances in Neural Information Processing Systems},
  volume={30},
  year={2017}
}

@article{velivckovic2017graph,
  title={Graph attention networks},
  author={Veli{\v{c}}kovi{\'c}, Petar and Cucurull, Guillem and Casanova, Arantxa and Romero, Adriana and Lio, Pietro and Bengio, Yoshua},
  journal={arXiv preprint arXiv:1710.10903},
  year={2017}
}

@article{zambaldi2018relational,
  title={Relational deep reinforcement learning},
  author={Zambaldi, Vinicius and Raposo, David and Santoro, Adam and Bapst, Victor and Li, Yujia and Babuschkin, Igor and Tuyls, Karl and Reichert, David and Lillicrap, Timothy and Lockhart, Edward and others},
  journal={arXiv preprint arXiv:1806.01830},
  year={2018}
}

@article{jiang2018graph,
  title={Graph convolutional reinforcement learning},
  author={Jiang, Jiechuan and Dun, Chen and Huang, Tiejun and Lu, Zongqing},
  journal={arXiv preprint arXiv:1810.09202},
  year={2018}
}

@article{chen2020gama,
  title={Gama: Graph attention multi-agent reinforcement learning algorithm for cooperation},
  author={Chen, Haoqiang and Liu, Yadong and Zhou, Zongtan and Hu, Dewen and Zhang, Ming},
  journal={Applied Intelligence},
  volume={50},
  pages={4195--4205},
  year={2020},
  publisher={Springer}
}

@article{li2017diffusion,
  title={Diffusion convolutional recurrent neural network: Data-driven traffic forecasting},
  author={Li, Yaguang and Yu, Rose and Shahabi, Cyrus and Liu, Yan},
  journal={arXiv preprint arXiv:1707.01926},
  year={2017}
}

@inproceedings{zeng2021graphlight,
  title={GraphLight: graph-based reinforcement learning for traffic signal control},
  author={Zeng, Zheng},
  booktitle={2021 IEEE 6th International Conference on Computer and Communication Systems (ICCCS)},
  pages={645--650},
  year={2021},
  organization={IEEE}
}

@article{yan2023graph,
  title={Graph cooperation deep reinforcement learning for ecological urban traffic signal control},
  author={Yan, Liping and Zhu, Lulong and Song, Kai and Yuan, Zhaohui and Yan, Yunjuan and Tang, Yue and Peng, Chan},
  journal={Applied Intelligence},
  volume={53},
  number={6},
  pages={6248--6265},
  year={2023},
  publisher={Springer}
}

@article{wang2020stmarl,
  title={STMARL: A spatio-temporal multi-agent reinforcement learning approach for cooperative traffic light control},
  author={Wang, Yanan and Xu, Tong and Niu, Xin and Tan, Chang and Chen, Enhong and Xiong, Hui},
  journal={IEEE Transactions on Mobile Computing},
  volume={21},
  number={6},
  pages={2228--2242},
  year={2020},
  publisher={IEEE}
}

@article{miller1963settings,
  title={Settings for fixed-cycle traffic signals},
  author={Miller, Alan J},
  journal={Journal of the Operational Research Society},
  volume={14},
  number={4},
  pages={373--386},
  year={1963},
  publisher={Taylor \& Francis}
}

@article{cools2013self,
  title={Self-organizing traffic lights: A realistic simulation},
  author={Cools, Seung-Bae and Gershenson, Carlos and D’Hooghe, Bart},
  journal={Advances in applied self-organizing systems},
  pages={45--55},
  year={2013},
  publisher={Springer}
}

@article{mannion2016experimental,
  title={An experimental review of reinforcement learning algorithms for adaptive traffic signal control},
  author={Mannion, Patrick and Duggan, Jim and Howley, Enda},
  journal={Autonomic road Transport Support Systems},
  pages={47--66},
  year={2016},
  publisher={Springer}
}

@article{haydari2020deep,
  title={Deep reinforcement learning for intelligent transportation systems: A survey},
  author={Haydari, Ammar and Y{\i}lmaz, Yasin},
  journal={IEEE Transactions on Intelligent Transportation Systems},
  volume={23},
  number={1},
  pages={11--32},
  year={2020},
  publisher={IEEE}
}

@article{abdulhai2003reinforcement,
  title={Reinforcement learning for true adaptive traffic signal control},
  author={Abdulhai, Baher and Pringle, Rob and Karakoulas, Grigoris J},
  journal={Journal of Transportation Engineering},
  volume={129},
  number={3},
  pages={278--285},
  year={2003},
  publisher={American Society of Civil Engineers}
}

@article{genders2016using,
  title={Using a deep reinforcement learning agent for traffic signal control},
  author={Genders, Wade and Razavi, Saiedeh},
  journal={arXiv preprint arXiv:1611.01142},
  year={2016}
}

@article{van2016deep,
  title={Deep reinforcement learning for coordination in traffic light control},
  author={Van Der Pol, Elise},
  journal={Master's thesis, University of Amsterdam},
  year={2016}
}

@inproceedings{behrendt2017deep,
  title={A deep learning approach to traffic lights: Detection, tracking, and classification},
  author={Behrendt, Karsten and Novak, Libor and Botros, Rami},
  booktitle={2017 IEEE International Conference on Robotics and Automation (ICRA)},
  pages={1370--1377},
  year={2017},
  organization={IEEE}
}

@inproceedings{behrisch2011sumo,
  title={SUMO--simulation of urban mobility: an overview},
  author={Behrisch, Michael and Bieker, Laura and Erdmann, Jakob and Krajzewicz, Daniel},
  booktitle={Proceedings of SIMUL 2011, The Third International Conference on Advances in System Simulation},
  year={2011},
  organization={ThinkMind}
}

@article{schulman2015high,
  title={High-dimensional continuous control using generalized advantage estimation},
  author={Schulman, John and Moritz, Philipp and Levine, Sergey and Jordan, Michael and Abbeel, Pieter},
  journal={arXiv preprint arXiv:1506.02438},
  year={2015}
}

@article{niri2021machine,
  title={Machine learning for optimised and clean Li-ion battery manufacturing: Revealing the dependency between electrode and cell characteristics},
  author={Niri, Mona Faraji and Liu, Kailong and Apachitei, Geanina and Ramirez, Luis Roman and Lain, Michael and Widanage, Dhammika and Marco, James},
  journal={Journal of Cleaner Production},
  volume={324},
  pages={129272},
  year={2021},
  publisher={Elsevier}
}

@article{jiang2004structural,
  title={A structural approach to the model generalization of an urban street network},
  author={Jiang, Bin and Claramunt, Christophe},
  journal={GeoInformatica},
  volume={8},
  pages={157--171},
  year={2004},
  publisher={Springer}
}

@article{devailly2021ig,
  title={IG-RL: Inductive graph reinforcement learning for massive-scale traffic signal control},
  author={Devailly, Fran{\c{c}}ois-Xavier and Larocque, Denis and Charlin, Laurent},
  journal={IEEE Transactions on Intelligent Transportation Systems},
  volume={23},
  number={7},
  pages={7496--7507},
  year={2021},
  publisher={IEEE}
}

@article{mousavi2017traffic,
  title={Traffic light control using deep policy-gradient and value-function-based reinforcement learning},
  author={Mousavi, Seyed Sajad and Schukat, Michael and Howley, Enda},
  journal={IET Intelligent Transport Systems},
  volume={11},
  number={7},
  pages={417--423},
  year={2017},
  publisher={Wiley Online Library}
}

@IEEEtranBSTCTL{IEEEexample:BSTcontrol,
  CTLuse_article_number     = "yes",
  CTLuse_paper              = "yes",
  CTLuse_forced_etal        = "no",
  CTLmax_names_forced_etal  = "10",
  CTLnames_show_etal        = "2",
  CTLuse_alt_spacing        = "yes",
  CTLalt_stretch_factor     = "4",
  CTLdash_repeated_names    = "yes",
  CTLname_format_string     = "{f.~}{vv~}{ll}{, jj}",
  CTLname_latex_cmd         = "",
  CTLname_url_prefix        = "[Online]. Available:"
}
\end{refcontext}


\end{document}